\documentclass[a4paper,UKenglish]{lipics-v2018-nopage}
\usepackage{microtype}

\newcommand{\IP}[1]{\mathrm{IP}[#1,\mathrm{qpoly}]} 
\newcommand{\abs}[1]{\vert #1 \vert}
\newcommand{\gen}[1]{\langle #1 \rangle}
\newcommand{\ket}[1]{\vert #1 \rangle}

\newcommand{\Int}{\mathbb Z}

\newcommand{\mybar}[1]{\lambda}

\newcommand{\Gg}{\mathscr{G}}

\newcommand{\poly}{\mathrm{poly}}

\newcommand{\Aa}{\mathscr{A}}

\newenvironment{proof-sketch}{\trivlist\item[]\emph{Brief proof sketch}:}%
{\unskip\nobreak\hskip 1em plus 1fil\nobreak$\Box$
\parfillskip=0pt%
\endtrivlist}

\title{Interactive Proofs with Polynomial-Time Quantum Prover for Computing the Order of Solvable Groups}
\titlerunning{Interactive Proofs for Order of Solvable Groups}

\author{Fran{\c c}ois Le Gall}{Graduate School of Informatics, Kyoto University\\
Yoshida-Honmachi, Sakyo-ku, Kyoto 606-8501, Japan
}{}{}{}

\author{Tomoyuki Morimae}{Yukawa Institute for Theoretical Physics, Kyoto University\\
Kitashirakawa Oiwakecho, Sakyo-ku, Kyoto 606-8502, Japan
}{}{}{}

\author{Harumichi Nishimura}{Graduate School of Informatics, Nagoya University\\
Chikusa-ku, Nagoya, Aichi 464-8601, Japan 
}{}{}{}

\author{Yuki Takeuchi}{NTT Communication Science Laboratories, NTT Corporation\\
3-1 Morinosato-Wakamiya, Atsugi, Kanagawa 243-0198, Japan \\\vspace{2mm}
Graduate School of Engineering Science, Osaka University\\
1-3 Machikaneyama-cho, Toyonaka, Osaka 560-8531, Japan
}{}{}{}

\authorrunning{F. Le Gall, T. Morimae, H. Nishimura and Y. Takeuchi}

\Copyright{Fran{\c c}ois Le Gall, Tomoyuki Morimae, Harumichi Nishimura and Yuki Takeuchi
}

\subjclass{}


\category{}

\relatedversion{}

\supplement{}

\funding{}

\EventEditors{John Q. Open and Joan R. Access}
\EventNoEds{2}
\EventLongTitle{42nd Conference on Very Important Topics (CVIT 2016)}
\EventShortTitle{CVIT 2016}
\EventAcronym{CVIT}
\EventYear{2016}
\EventDate{December 24--27, 2016}
\EventLocation{Little Whinging, United Kingdom}
\EventLogo{}
\SeriesVolume{42}
\ArticleNo{23}
\nolinenumbers 
\hideLIPIcs  

\begin{document}

\maketitle

\begin{abstract}
In this paper we consider what can be computed by a user interacting with a potentially malicious server, when the server performs polynomial-time quantum computation but the user can only perform polynomial-time classical (i.e., non-quantum) computation. Understanding the computational power of this model, which corresponds to polynomial-time quantum computation that can be efficiently verified classically, is a well-known open problem in quantum computing. Our result shows that computing the order of a solvable group, which is one of the most general problems for which quantum computing exhibits an exponential speed-up with respect to classical computing, can be realized in this model. 
 \end{abstract}

\section{Introduction}
First-generation quantum computers will be implemented in the ``cloud'' style, since only few groups, such as governments or huge companies, will be able to possess such expensive and high-maintenance machines. In fact, IBM has recently opened their 16-qubit machine for a cloud service \cite{IBM}. In a future when many companies provide their own quantum cloud computing services, a malicious company might emerge who is trying to palm a user off with a wrong result forged from their fake quantum computer. In addition, even if a fortunate user is interacting with a honest server, some noises in the server's gate operations might change the result. How can a user verify the correctness of the server's quantum computation? If the user has his/her own quantum computer, the user can of course check the server's result, but in this case the user may not need the cloud service in the first place. If the solution of the problem is easily verifiable (e.g., integer factoring), the user can naturally verify the correctness of the server's result, but many problems considered in quantum computing are not believed to have this property. Verifying classically and efficiently a server's quantum computation is indeed in general highly nontrivial. 

It is known that if at least two servers, who are entangled but not communicating with each other, are allowed, then any problem solvable in quantum polynomial time can be verified by a classical polynomial-time user who exchanges classical messages with the servers~\cite{Ji,MattMBQC,RUV}. However, the assumption that servers are not communicating with each other is somehow unrealistic: how can the user guarantee that remote servers are not communicating with each other? 

Whether the number of the servers can be reduced to one is a well-known open problem~\cite{AharonovVazirani}. For certain computational problems solvable in quantum polynomial time, it is known that this can be done. 
Simon's problem~\cite{Simon} and factoring~\cite{Shor} are trivial examples, since the answer can be directly checked in classical polynomial time. It is known that recursive Fourier sampling~\cite{BV}, which was the first problem that separates efficient quantum and classical computing, can be verified by a polynomial number of message exchanges with a single quantum server~\cite{MattFourier}. Moreover, it was shown that certain promise problems related to quantum circuits in the second level of the Fourier  hierarchy~\cite{FH} are verifiable by a classical polynomial-time user interacting with a single quantum server who sends only a single message to the user~\cite{Tommaso,FH2}. 
\vspace{2mm}

\noindent
{\bf Our results.}
In this paper we consider the problem of computing the order, i.e., the number of elements, of a finite group given as a black-box group (the concept of black-box groups is defined in Section \ref{sec:prelim}). This problem is central in computational group theory, especially since the ability of computing the order makes possible to decide membership in subgroups. This problem has also been the subject of several investigations in computational complexity \cite{Aaronson+07,Babai92,Babai+STOC09,Babai+FOCS84,WatrousFOCS00,Watrous01}. The seminal result by Babai \cite{Babai92}, especially, which put this problem in the complexity class AM, has been one of the fundamental motivations behind the concept of interactive proofs. Note that this is clearly a hard problem for classical computation: it is easy to show that no polynomial-time classical algorithm exists in the black-box setting, even if the input is an abelian group \cite{Babai+FOCS84}.  

Most of the known quantum algorithms that achieve exponential speedups with respect to the best known classical algorithms are for group-theoretic problems, and especially problems over abelian groups. Shor's algorithm for factoring \cite{Shor}, for instance, actually computes the order of a cyclic black-box group. Watrous has shown that the group order problem can be solved in quantum polynomial time when the input group is solvable \cite{Watrous01}. Since the class of solvable groups, defined in Section~\ref{sec:prelim}, is a large\footnote{It is known (see for instance \cite{Blackburn+07}) that 
\[
\lim_{m\to\infty }\frac{\log \Gg_s(m)}{\log \Gg(m)}=1,
\]
where $\Gg(m)$ denotes the number of finite groups of order at most $m$ and $\Gg_s(m)$ denotes the number of finite solvable groups of order at most $m$. It is even conjectured that the quotient $\Gg_s(m)/\Gg(m)$ goes to~$1$ when $m$ goes to infinity, i.e., most finite groups are solvable.} class of finite groups that includes all abelian groups, this result significantly generalized Shor's algorithm. Watrous' algorithm can actually be seen as one of the most general results achieving an exponential speedup with respect to classical computation.

In this paper we show that the group order problem over solvable groups is also verifiable with a single server. More formally, in Section \ref{sec:prelim}, where we introduce the relevant model of interactive protocols, we will introduce the notation $\IP{k}$ to denote the class of computational problems that are verifiable by a classical polynomial-time user interacting in $k$ messages with a server who works in quantum polynomial time. Our main result is as follows. 
\begin{theorem}\label{th:main}
The solvable group order problem is in the complexity class $\IP{3}$. Moreover, if the set of prime factors of the order is also given as input, then the solvable group order problem is in $\IP{2}$.
\end{theorem}
This result shows that for this important computational problem, the number of servers can be reduced to one as well, using a small number of messages. 
Note that assuming, in the second part of Theorem \ref{th:main}, that the set of prime factors of the order is known corresponds to several practical situations. An important example is computing the order of $p$-groups\footnote{A (finite) $p$-group, where $p$ is a prime, is a group of order $p^r$ for some integer $r\ge 1$. A basic result from group theory shows that any $p$-group is solvable.} with~$p$ known, which cannot be done in polynomial time in the classical setting \cite{Babai+FOCS84}. 
The main open question is whether the number of messages can also be reduced to 2 without any assumption on the prime factors. 

\vspace{2mm}

\noindent
{\bf Other related works.}
In addition to the introduction of multiple servers mentioned above, there are other approaches considered in the literature for constructing verification systems for quantum computation.

First, if the user is allowed to be ``slightly quantum'', any problem solvable in quantum polynomial time can be efficiently verified with a single quantum server. For example, Refs.~\cite{Aharonov,FK} assume that the user can generate randomly-rotated single-qubit states, and Refs.~\cite{posthoc,HM,MNS} assume that the user can measure single-qubit states. 

Second, since the class BQP (the class of decision problems that can be solved in quantum polynomial-time) is trivially in PSPACE and ${\rm PSPACE}={\rm IP}$ \cite{Lund+92,Shamir92}, any problem in BQP can be classically verified using generic interactive proof protocols for PSPACE. In such protocols, however, the server has unbounded computational power. A tempting approach is to try to specialize these generic protocols to the class BQP, with the hope that the server's necessary computational power may be reduced. Ref.~\cite{AharonovGreen} made an significant first step in this direction.

Finally, it has been shown very recently that assuming that the learning with errors problem is intractable for polynomial-time quantum computation, any problem solvable in quantum polynomial time can be efficiently verified with a single quantum server and a single classical user~\cite{Mahadev}.
\section{Preliminaries}\label{sec:prelim}
In this paper we assume that the reader is familiar with the standard notions of group theory (we refer to, e.g., \cite{Isaacs08} for a good introduction). All the groups considered will be finite. Given a group $G$, we use $\abs{G}$ to denote its order (i.e., the number of elements in~$G$), and use $e$ to denote its identity element. 
Given elements $g_1,\ldots, g_r\in G$, we denote $\gen{g_1,\ldots, g_r}$ the subgroup of $G$ generated by $g_1,\ldots, g_r$.\vspace{2mm}

\noindent
{\bf Black-box groups.}
We now describe the model of black-box groups. This concept, in which each group element is represented by a string and each group operation is implemented using an oracle, was first introduced by Babai and Szemer{\'e}di \cite{Babai+FOCS84} to describe group-theoretic algorithms in the most general way, without having to concretely specify how the elements are represented and how groups operations are implemented. Indeed, any efficient algorithm in the black-box group model gives rise to an efficient concrete algorithm whenever the oracle operations can be replaced by efficient procedures. Especially, performing group operations can be done directly on the elements in polynomial time for many natural groups, including permutation groups and matrix groups where the group elements are represented by permutations and matrices, respectively. In the quantum setting, black-box groups have first been considered by Ivanyos et al.~\cite{Ivanyos+03} and Watrous~\cite{WatrousFOCS00,Watrous01}. 

A black-box group is a representation of a group $G$ where each element of $G$ is uniquely encoded by a binary string of a fixed length $n$, which is called the encoding length. The encoding length $n$ is known. In order to be able to express the complexity of black-box group algorithms in terms of the group order $|G|$, and not in terms of the encoding length, we make the standard assumption that $n=O(\log |G|)$. Oracles are available to perform group operations. More precisely, two oracles are available. A first oracle performs the group product: given two strings representing two group elements $g$ and $h$, the oracle outputs the string representing $gh$. The second oracle performs inversion: given a string representing an element $g\in G$, the oracle outputs the string representing the element $g^{-1}$. Note that the two oracles may behave arbitrarily on strings not corresponding to elements in $G$; this is not a problem since our protocols will never use the oracles on such strings. We say that a group~$G$ is input as a black-box if a set of strings representing generators $\{g_1,\ldots,g_s\}$ of~$G$ with $s=O(\log\abs{G})$ is given as input and queries to the oracles can be done at cost 1.\footnote{The assumption $s=O(\log\abs{G})$ is standard. Indeed, every group~$G$ has a generating set of size $O(\log \abs{G})$. Additionally, a set of generators of any size can be converted efficiently into a set of generators of size $O(\log\abs{G})$ by taking random products of elements~\cite{BabaiSTOC91}.} The input length is thus $sn=\poly(\log |G|)$.
 
To be able to take advantage of the power of quantum computation when dealing with black-box groups, the oracles performing the group operations have to be able to deal with quantum superpositions. Concretely, this is done as follows (see \cite{Ivanyos+03,WatrousFOCS00,Watrous01}). Let $s\colon G\to \{0,1\}^n$ denote the encoding of elements as binary strings. We assume that a quantum oracle $V_G$ is available, such that $V_G(\ket{s(g)}\ket{s(h)})=\ket{s(g)}\ket{s(gh)}$ for any two elements $g,h\in G$, and behaving in an arbitrary way on other inputs (i.e., strings not in $s(G)$). Another quantum oracle $V'_G$ is also available, such that $V'_G(\ket{s(g)}\ket{s(h)})=\ket{s(g)}\ket{s(g^{-1}h)}$ for any $g, h\in G$ and again behaving in an arbitrary way on other inputs.\vspace{2mm}

\noindent
{\bf Approximate sampling in black-box groups.}
Babai  \cite{BabaiSTOC91} proved the following result for general groups, which shows that elements of a black-box group can be efficiently sampled nearly uniformly.
\begin{theorem}{(\cite{BabaiSTOC91})}\label{th_babai}
Let $G$ be a black-box group. For any $\varepsilon>0$,
there exists a classical randomized algorithm running
in time polynomial in $\log(|G|)$ and $\log(1/\varepsilon)$ that outputs an element of $G$ such that each $g\in G$ is output with probability in range $(1/|G|-\varepsilon,1/|G|+\varepsilon)$.
\end{theorem}

\noindent
{\bf Solvable groups.}
Before discussing solvable groups, let us introduce the following concept of polycyclic generating sequences (see \cite{Holt+05} for details). 
\begin{definition}\label{def:solvable}
\label{def-sg}
Let $G$ be a group. A polycyclic generating sequence of $G$ is a sequence $(h_1,\ldots,h_t)$ of $t$ elements from $G$, for some integer $t$, such that:
\begin{enumerate}
\item
$\gen{h_1,\ldots,h_t}=G$;
\item
for each $1< j\le t$, the subgroup $\gen{h_1,\ldots,h_{j-1}}$ is normal in $\gen{h_1,\ldots,h_j}$.
\end{enumerate} 
\end{definition} 

There are many equivalent definitions of solvable groups in the literature (see, e.g., \cite{Holt+05} for a thorough discussion). In this paper we will use the following characterization: a finite group is solvable if and only if it has a polycyclic generating sequence. This characterization, which was already used by Watrous~\cite{Watrous01}, is the most convenient for our purpose. As discussed in \cite{Watrous01}, for any finite solvable group $G$ given as a black box, a polycyclic generating sequence $(h_1,\ldots,h_t)$ with $t=O(\log\abs{G})$ can be computed classically in polynomial time with high probability using for instance the randomized algorithm by Babai et al.~\cite{Babai+JCSS95}. 

Watrous showed that the order of a solvable black-box group can be computed in polynomial time in the quantum setting. We state this result in the following theorem.

\begin{theorem}{(\cite{Watrous01})}\label{th_watrous-order}
Let $G$ be a solvable group given as a black-box group. There exists a quantum algorithm running in time $\poly(\log\abs{G})$ that outputs $|G|$ with probability at least $1-1/\poly(|G|)$.
\end{theorem}\vspace{2mm}

Let $G$ be a solvable group and $(h_1,\ldots,h_t)$ be a polycyclic generating sequence of $G$. In the following we will write $H_j=\gen{h_1,\ldots,h_j}$ for each $j\in\{1,\ldots,t\}$, and for convenience write $H_0=\{e\}$. Since $H_{j}$ is obtained from $H_{j-1}$ by adding one generator, the factor group $H_{j}/H_{j-1}$ is cyclic. Let us write its order $m_j$. Note that the order of $G$ is thus the product $m_1m_2\cdots m_t$. A fundamental (and easy to show) property of polycyclic generating sequences is the following:
For any $j\in\{1,\ldots, t\}$, any element $h\in H_j$ can be written, in a unique way, as $h=h_1^{a_1}h_{2}^{a_{2}}\cdots h_{j}^{a_{j}}$ with integers $a_i\in\{0,1,\ldots,m_i-1\}$ for $i\in\{1,\ldots,j\}$. We call this sequence $(a_1,\ldots,a_{j})$ the decomposition of $h$ over $H_j$. Watrous \cite{Watrous01} showed that 
the decomposition of any element can be computed efficiently in the quantum setting, which immediately leads to an efficient algorithm for membership testing in the subgroups $H_{j}$. We state these two results, separately, in the following theorem.

\begin{theorem}{(\cite{Watrous01})}\label{th_watrous}
Let $G$ be a solvable group given as a black-box group and let $(h_1,\ldots,h_t)$ be a polycyclic generating sequence of $G$ with $t=O(\log\abs{G})$. There exist two quantum algorithms $\Aa_1$ and $\Aa_2$ running in time polynomial in $\log\abs{G}$ as follows.
\begin{itemize}
\item 
Algorithm $\Aa_1$ receives an integer $j\in\{1,\ldots t\}$ and an element $h\in H_j$, and outputs with probability at least $1-1/\poly(|G|)$ the decomposition of $h$ over $H_{j}$.
\item
Algorithm $\Aa_2$ receives an integer $j\in\{1,\ldots t\}$ and an element $h\in G$, and decides whether $h\in H_j$ or not. The decision is correct with probability at least $1-1/\poly(|G|)$.
 \end{itemize}
\end{theorem}\vspace{2mm}

\noindent
{\bf Interactive proofs with efficient quantum prover.}
Interactive proof systems are typically described as protocols for decision problems. In this paper it will be more convenient to consider interactive proofs for computing functions, since we are interesting in computing the order of the input group.\footnote{In order to be completely rigorous, we should actually define this concept for functional problems where  the input is represented using oracles (since we are dealing with black-box groups where the group operation is represented by oracles). We nevertheless omit this purely technical point in the exposition.} The definition we give below is inspired by \cite{Goldwasser+ITCS18}.

Let $f: X\to\{0,1\}^\ast$ be a function, where $X$ is a finite set. We consider protocols between a prover and a verifier, who both receives as input an element $x\in X$ and can exchange classical messages of polynomial length. At the end of the protocol, the verifier outputs either some $y\in\{0,1\}^\ast$ or one special element $\perp$. We say that the function $f$ has a \emph{$k$-message polynomial-time interactive proof} if there exists a $k$-message protocol in which the verifier works in classical polynomial time, such that the following properties hold:
\begin{enumerate}
\item
(completeness) there is a prover $P$ such that the verifier's output $y$ satisfies $y=f(x)$ with probability at least $2/3$ when interacting with $P$;
\item
(soundness) for any prover $P'$, the verifier's output $y$ satisfies $y\in\{f(x),\perp\}$ with probability at least $2/3$ when interacting with $P'$.
\end{enumerate}
The prover $P$ in the completeness condition is called the \emph{honest prover}.

The above definition makes no assumption on the computational powers of the provers. Our main definition is obtained by restricting the computational power of the \emph{honest} prover, i.e., the prover $P$ in the completeness condition. 
\begin{definition}\label{def:IP}
A function $f$ is in the class $\IP{k}$ if it has a $k$-message polynomial-time interactive proof where the honest prover $P$ works in quantum polynomial time.
\end{definition}
The notation $\IP{k}$ comes from its definition as a $k$-message interactive protocol with a prover working in quantum polynomial time (when honest).
We stress that in Definition \ref{def:IP} there is no assumption on the computational power of $P'$ for the soundness. 

\section{2-Message Protocol with Known Prime Factors}\label{sec:2m}
In this section we assume that the prime factors of the order of the black-box group $G$ are known. We present a 2-message protocol in this case, which proves the second part of Theorem \ref{th:main}.

\subsection{Preliminaries}
We will need the following result in our protocol.


\begin{theorem}\label{th:classical2}
Let $G$ be a solvable group given as a black-box group. Let $p_1,\ldots,p_\ell$ denote the prime factors of $|G|$ and assume that the set $S=\{p_1,\ldots,p_\ell\}$ is also given as input.
There exists a classical algorithm running in time polynomial in $\log\abs{G}$ that outputs elements $h_1,\ldots,h_t\in G$, with $t=\poly(\log |G|)$, and $t$ prime numbers $r_1,\ldots,r_t\in S$ such that, with probability at least $1-1/\poly(|G|)$, the following conditions hold:
\begin{itemize}
\item
$(h_1,\ldots,h_t)$ is a polycyclic generating sequence of $G$;
\item
the order of $H_{i}/H_{i-1}$ is either 1 or $r_i$ for each $1\le i \le t$, where we denote $H_i=\gen{h_1,\ldots,h_i}$ for $1\le i\le t$ and $H_0=\{e\}$.
\end{itemize}
\end{theorem}
Before proving Theorem \ref{th:classical2}, let us discuss the main idea of the algorithm in this theorem. The approach is to start with an arbitrary polycyclic generating sequence and refine it by replacing each element by decreasing powers of it. Consider for instance the cyclic group of order~12, for which we have $\ell=2$, $p_1=2$, $p_2=3$ and $|G|=12$. 
Assume that we start with the polycyclic generating sequence $(k_1)$ consisting of a unique element $k_1$ of order~12. We refine this sequence as $(h_1,h_2,h_3)$ with $h_1=k_1^{|G|/p_1}=k_1^6$, $h_2=k_1^{|G|/p_1^2}=k_1^3$ and $h_3=k_1^{|G|/(p_1^2p_2)}=k_1$. This is a polycyclic generating sequence with $|H_1/H_0|=2$, $|H_2/H_1|=2$ and $|H_3/H_2|=3$. The difficulty is that naturally we do not know the order $|G|$. Remember nevertheless that we know the encoding length $n$ of the black-box group, which is an upper bound on $\log_2|G|$. This means that the quantity $m=p_1^n\times \ldots\times p_\ell^n$ is a multiple of the order $|G|$, and thus we can use the same approach, working with $m$ instead of $|G|$ when refining the original polycyclic generating sequence.

\begin{proof}[Proof of Theorem \ref{th:classical2}]
Let us consider the function $\lambda\colon \{1,\ldots,\ell\}\times \{1,\ldots,n\}\to \Int $ such that 
\[
\lambda(i,a)=p_i^{n-a}\times p_{i+1}^{n}\times \cdots \times p_{\ell}^{n}
\] 
for any $(i,a)\in\{1,\ldots,\ell\}\times \{1,\ldots,n\}$. 
Now consider the sequence
\begin{equation}\label{eq:seq}
({\lambda(1,1)},\ldots,{\lambda(1,n)},{\lambda(2,1)},\ldots,{\lambda(2,n)},\ldots,{\lambda(\ell,1)},\ldots,{\lambda(\ell,n)})
\end{equation}
consisting of $\ell n$ integers (the integers in the sequence are strictly decreasing). Define the function $\mu\colon\{1,\ldots,\ell n\}\to \Int$ such that $\mu (j)$ is the $j$-th integer in Sequence (\ref{eq:seq}). 
Note that $\mu(j-1)/\mu(j)\in S$ for any $j\in\{2,\ldots,\ell n\}$. 

We now describe our algorithm that computes the claimed generating sequence. 

We first compute a polycyclic generating sequence $(k_1,\ldots,k_{t'})$ of $G$ with $t'=O(\log|G|)$ using the randomized polynomial-time algorithm from \cite{Babai+JCSS95}, already mentioned in Section \ref{sec:prelim}, which succeeds with probability at least $1-1/\poly(|G|)$. Let us write $K_{i'}=\gen{k_1,\ldots,k_{i'}}$ for each $1\le {i'}\le t'$, and $K_0=\{e\}$.

We now show how to refine the polycyclic generating sequence.  
For each ${i'}\in\{1,\ldots,t'\}$, we replace $k_{i'}$ by the sequence of $\ell n$ elements $(k_{i'}^{\mu(1)},\ldots k_{i'}^{\mu(\ell n)})$, which gives a new sequence 
\begin{equation}\label{eq:seq2}
\left(k_1^{\mu(1)},\ldots, k_1^{\mu(\ell n)},k_2^{\mu(1)},\ldots, k_2^{\mu(\ell n)},\ldots,k_{t'}^{\mu(1)},\ldots, k_{t'}^{\mu(\ell n)}
\right),
\end{equation}
of $\ell n t'$ elements. Sequence (\ref{eq:seq2}) is a polycyclic generating sequence of $G$ since $(k_1,\ldots,k_{t'})$ is a polycyclic generating sequence of $G$ and $\mu(\ell n)=1$. For any ${i'}\in\{1,\ldots,t'\}$, observe that
\begin{equation}\label{eq:seq3}
\left|
\gen{k_1^{\mu(1)},\ldots,k_{i'}^{\mu(j)}}/\gen{k_1^{\mu(1)},\ldots,k_{{i'}}^{\mu(j-1)}}
\right|
\in\{1,\mu(j-1)/\mu(j)\}
\end{equation}
for any $j\in\{2,\ldots \ell n\}$. Similarly for any ${i'}\in\{2,\ldots t'\}$ we have
\begin{equation}\label{eq:seq4}
\left|
\gen{k_1^{\mu(1)},\ldots,k_{i'}^{\mu(1)}}/\gen{k_1^{\mu(1)},\ldots,k_{{i'}-1}^{\mu(\ell n)}}
\right|
\in\{1,p_1\}.
\end{equation}

Let us rename the elements of Sequence (\ref{eq:seq2}) as $h_1,\ldots,h_t$, with $t=\ell n t'$. Note that $t=O(\ell(\log|G|)^2)=O((\log|G|)^3)$. Let us write $H_i=\gen{h_1,\ldots,h_i}$ for $1\le i\le t$ and $K_0=\{e\}$. For each $1\le i \le t$, the order of $H_{i}/H_{i-1}$ is either 1 or $r_i$, where $r_i$ can be determined from Equations~(\ref{eq:seq3}) and (\ref{eq:seq4}). More concretely,  $r_i$ is of the form $\mu(j-1)/\mu(j)$ for some $j$ (which can be immediately computed from $i$) when $H_{i}/H_{i-1}$ corresponds to the case of Equation~(\ref{eq:seq3}), and $r_i=p_1$ when $H_{i}/H_{i-1}$ corresponds to the case of Equation~(\ref{eq:seq4}). Note that in both cases we have $r_i\in S$, from the property $\mu(j-1)/\mu(j)\in S$ mentioned before.
\end{proof}

\subsection{The protocol}
 Let $S=\{p_1,\ldots, p_\ell\}$ denote the set of prime factors of $|G|$, which is given as an additional input. The protocol is given in Figure \ref{fig:protocol}. The main idea is that the verifier can, using Theorem~\ref{th:classical2}, compute by itself a polycyclic generating sequence $(h_1,\ldots,h_t)$ and prime numbers $r_1,\ldots,r_t$ such that $|H_{i}/H_{i-1}|\in\{1,r_i\}$ for each $1\le i \le t$. This is done at Step~1 of the protocol. Note that $|G|=\prod_{i=1}^t |H_i/H_{i-1}|$. The purpose of Steps 2-5 is to decide whether $|H_i/H_{i-1}|=1$ or $|H_i/H_{i-1}|=r_i$, for each $i\in\{1,\ldots,t\}$, by interacting with the prover. More precisely, the verifier interacts with the prover to test, for each $i$, whether $h_i\in H_{i-1}$ or $h_i\notin H_{i-1}$. This requires testing non-membership in a solvable group with a polynomial-time quantum prover, which is achieved by sending (at Step 3) to the prover the element $h_i^{s_i}x_i$ for a random bit $s_i$ and a random element $x_i$, and asking the prover to find the chosen bit $s_i$.
These tests enable the verifier to decide which of the two cases holds (at Steps~5.1 and 5.2), and then to compute $|G|$ at Step 6, or to detect cheating (at Step 5.3).

\begin{figure}[t!]
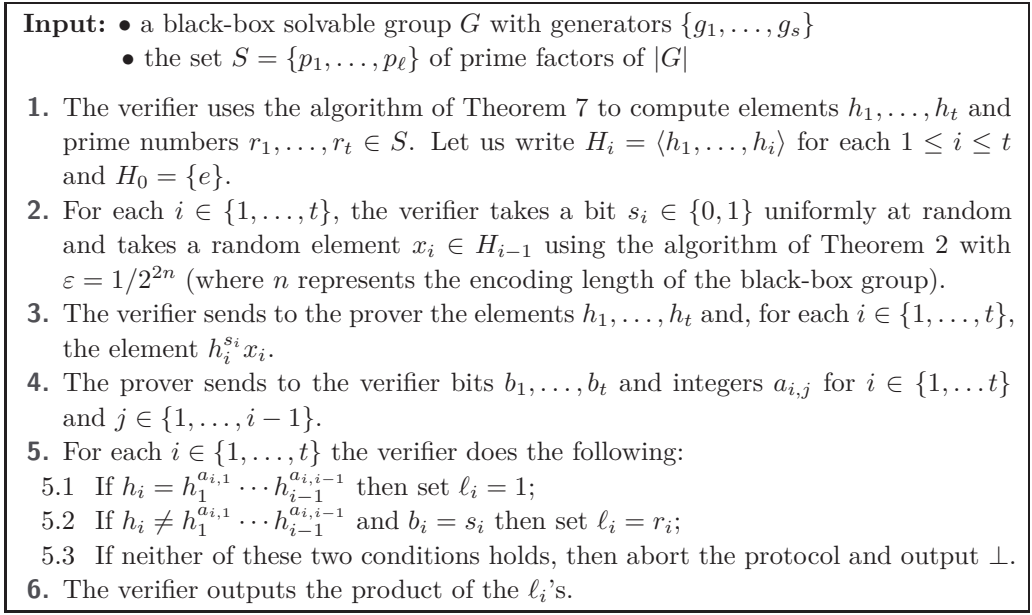

\begin{center}
\fbox{
\begin{minipage}{13 cm} 
{\bf Input:} $\bullet$ a black-box solvable group $G$ with generators $\{g_1,\ldots,g_s\}$\\
\phantom{Input:aa}$\bullet$ the set $S=\{p_1,\ldots, p_\ell\}$ of prime factors of $|G|$\vspace{2mm}
\begin{enumerate}
\item
The verifier uses the algorithm of Theorem \ref{th:classical2} to compute elements $h_1,\ldots,h_t$ and prime numbers $r_1,\ldots,r_t\in S$. Let us write $H_i=\gen{h_1,\ldots,h_i}$ for each $1\le i\le t$ and $H_0=\{e\}$.
\item
For each $i\in\{1,\ldots,t\}$, the verifier takes a bit $s_i\in\{0,1\}$ uniformly at random and takes a random element $x_i\in H_{i-1}$ using the algorithm of Theorem \ref{th_babai} with $\varepsilon=1/2^{2n}$ (where $n$ represents the encoding length of the black-box group).
\item
The verifier sends to the prover the elements $h_1,\ldots,h_t$ and, for each $i\in\{1,\ldots,t\}$, the element $h_i^{s_i}x_i$.
\item
The prover sends to the verifier bits $b_1,\ldots,b_t$ and integers $a_{i,j}$ for $i\in\{1,\ldots t\}$ and $j\in\{1,\ldots,i-1\}$.
\item
For each $i\in\{1,\ldots,t\}$ the verifier does the following:
\begin{itemize}
\item[5.1]
If $h_i=h_1^{a_{i,1}}\cdots h_{i-1}^{a_{i,i-1}}$ then set $\ell_i=1$;
\item[5.2]
If $h_i\neq h_1^{a_{i,1}}\cdots h_{i-1}^{a_{i,i-1}}$ and $b_i=s_i$ then set $\ell_i=r_i$;
\item[5.3] 
If neither of these two conditions holds, then abort the protocol and output $\perp$.
\end{itemize}
\item
The verifier outputs the product of the $\ell_i$'s.
\end{enumerate}
\end{minipage}
}
\end{center}\vspace{-4mm}
\caption{Our 2-message protocol computing the order of a solvable group when the prime factors of the order are known.}\label{fig:protocol}
\end{figure}

\subsection{Analysis of the protocol}
We now analyze the protocol of Figure \ref{fig:protocol}.  Let $h_1,\ldots,h_t$ be the group elements and $r_1,\ldots,r_t\in S$ be the prime numbers computed at Step 1. The analysis below is done under the assumption that $(h_1,\ldots,h_t)$ is a polycyclic generating sequence of $G$ and $|H_i/H_{i-1}|\in\{1,r_i\}$ for all $i\in\{1,\ldots,t\}$, which is true with probability at least $1-1/\poly(|G|)$ from Theorem \ref{th:classical2}. 

Let us first consider the correctness, i.e., showing that there exists a prover (working in quantum polynomial time) who makes the verifier able to compute $|G|$ with high probability.  This prover acts as follows. For each $i\in\{1,\ldots,t\}$, the prover checks if the element $h_i^{s_i}x_i$ received at Step 3 is in the subgroup $H_{i-1}$, using Algorithm $\Aa_2$ of Theorem~\ref{th_watrous}. If the prover learns that this element is in $H_{i-1}$ then the prover applies Algorithm $\Aa_1$ of Theorem \ref{th_watrous} to obtain a decomposition $(a_{i,1},\ldots,a_{i,i-1})$ of $h_i$ over $H_{i-1}$, and sends to the verifier the bit $b_i=0$ and these values $a_{i,1},\ldots,a_{i,i-1}$. If the prover learns that this element is not in $H_{i-1}$, then the prover sends to the verifier the bit $b_i=1$ and arbitrary values $a_{i,1},\ldots,a_{i,i-1}$.

Let us analyze the verifier's output when interacting with the above prover. If $|H_i/H_{i-1}|=1$ then we have $h_i\in H_{i-1}$ and thus $h_i^{s_i}x_i\in H_{i-1}$ whatever the value of $s_i$ is. With probability at least $1-1/\poly(|G|)$, the prover's message is thus $b_i=0$ and $a_{i,1},\ldots,a_{i,i-1}$ corresponding to the decomposition of $h_i$ over $H_{i-1}$, and then the verifier sets $\ell_i=1$. If $|H_i/H_{i-1}|=r_i$ then we have $h_i\notin H_{i-1}$ and thus $h_i^{s_i}x_i\in H_{i-1}$ if and only if $s_i=0$. With probability at least $1-1/\poly(|G|)$, the bit $b_i$ sent by the prover satisfies $b_i=s_i$, and thus the verifier sets $\ell_i=r_i$ (since the second part of the message $a_{i,1},\ldots,a_{i,i-1}$ cannot correspond to the decomposition of $h_i$ over $H_{i-1}$). In conclusion, with probability at least $1-1/\poly(|G|)$ the output at Step 6 is
\[
\prod_{i=1}^t \ell_i=\prod_{i=1}^t |H_i/H_{i-1}| =|G|.
\] 

Let us now consider the soundness, i.e., showing that for any prover the verifier outputs either $|G|$ or $\perp$ with high probability. It is clear that if $|H_i/H_{i-1}|=r_i$, then the prover cannot convince the verifier to set $\ell_i=1$, since there is no set of integers $a_{i,1},\ldots,a_{i,i-1}$ such that $h_i=h_1^{a_{i,1}}\cdots h_{i-1}^{a_{i,i-1}}$. On the other hand, if $|H_i/H_{i-1}|=1$ then the prover cannot convince the verifier to set $\ell_i=r_i$ unless the prover is able to decide whether $s_i=0$ or $s_i=1$ from the element $h_i^{s_i}x_i$ received, which cannot be done 
with probability larger than $\frac{1}{2}+\frac{1}{2}\delta$, where 
\[
\delta=\frac{1}{2}\sum_{h\in H_{i-1}}\left|\Pr_{x_i}[x_i=h]-\Pr_{x_i}[h_ix_i=h]\right|
\] 
represents the variational distance between the two probability distributions $x_i$ and $h_ix_i$ (seen as distributions over $H_{i-1}$).
We have
\begin{align*}
\delta&\le 
\frac{1}{2}\sum_{h\in H_{i-1}}\left|\Pr[x_i=h]-\frac{1}{|H_{i-1}|}\right|+
\frac{1}{2}\sum_{h\in H_{i-1}}\left|\Pr[h_ix_i=h]-\frac{1}{|H_{i-1}|}\right|\\
&=
\frac{1}{2}\sum_{h\in H_{i-1}}\left|\Pr[x_i=h]-\frac{1}{|H_{i-1}|}\right|+
\frac{1}{2}\sum_{h\in H_{i-1}}\left|\Pr[x_i=h_i^{-1}h]-\frac{1}{|H_{i-1}|}\right|\\
&=
\sum_{h\in H_{i-1}}\left|\Pr[x_i=h]-\frac{1}{|H_{i-1}|}\right|\\
&\le |H_{i-1}|\varepsilon\\
&\le 1/2^n,
\end{align*}
where the second inequality follows from Theorem \ref{th_babai} and the third inequality follows from our choice of $\varepsilon$ and the upper bound $|G|\le 2^n$.
Thus, for any fixed $i$ such that $|H_i/H_{i-1}|=1$, the prover cannot convince the verifier to set $\ell_i=r_i$ with probability greater than $\frac{1}{2}+\frac{1}{2^{n+1}}=1/2+1/\poly(|G|)$. Let us now bound the probability that the verifier's output is either $|G|$ or $\perp$. This corresponds to the probability that the verifier does not output an integer different from the order of~$G$. Note that the verifier can output an integer not equal to the order only if 
the prover forces the verifier to set $\ell_i\neq |H_i/H_{i-1}|$ for at least one index $i$.
From the above analysis, we know that this can happen with probability at most $1/2+1/\poly(|G|)$, i.e., such a cheating is detected by the verifier at Step 5.3 with probability at least $1/2-1/\poly(|G|)$, in which case the verifier immediately aborts the protocol and outputs $\perp$.  
Thus the overall probability that the verifier's output is either $|G|$ or $\perp$ is at least $1/2-1/\poly(|G|)$. Note finally that this probability can be amplified to reach the soundness threshold of $2/3$ used in Definition \ref{def:IP} by repeating the protocol of Figure \ref{fig:protocol} a constant number of times in parallel and deciding the output based on a standard threshold argument.

\section{General 3-Message Protocol}\label{sec:3m}
In this section we show that when the prime factors of the order of $G$ are not known, we can design a 3-message protocol, which proves the first part of Theorem \ref{th:main}.

\subsection{The protocol}
Our 3-message protocol, described in Figure \ref{fig:protocol2}, is obtained by modifying the protocol of the previous section. More precisely, Step 1 in the protocol of the previous section is replaced by two steps (Steps 0 and 1 in Figure \ref{fig:protocol2}): instead of having the verifier compute a polycyclic generating sequence $(h_1,\ldots,h_t)$ using Theorem \ref{th:classical2}, which requires the knowledge of the set of factors of $|G|$, in the new protocol the prover computes by itself this sequence and sends it at Step 0 to the verifier, who then checks that the sequence is really correct at Step 1. All the other steps 2-6 are exactly the same as for the protocol in Figure~\ref{fig:protocol} (one small exception is Step 3, which is slightly rewritten since the  polycyclic generating sequence does not need to be sent to the prover anymore). 
 
\begin{figure}[t!]
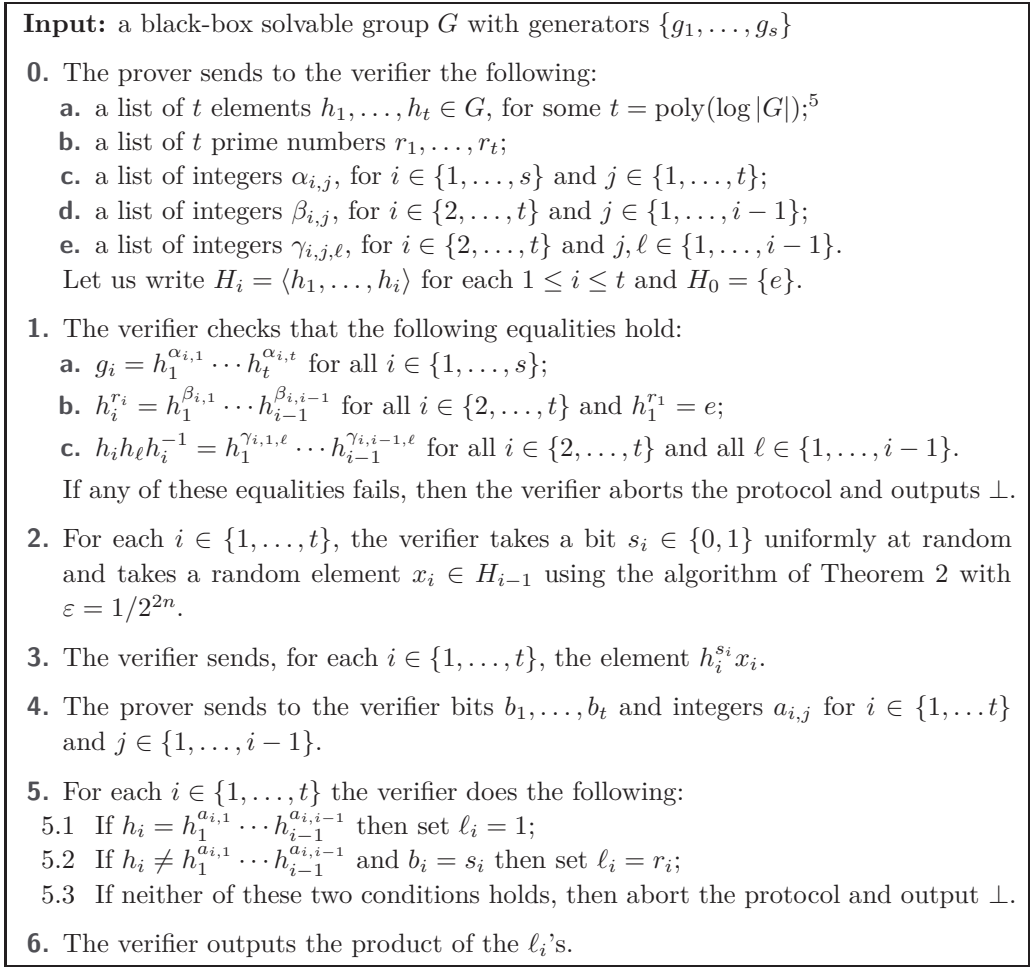

\begin{center}
\fbox{
\begin{minipage}{13 cm} 
{\bf Input:} a black-box solvable group $G$ with generators $\{g_1,\ldots,g_s\}$\vspace{2mm}
\begin{enumerate}
\setcounter{enumi}{-1}
\item
The prover sends to the verifier the following:
\begin{enumerate}
\item
a list of $t$ elements $h_1,\ldots,h_t\in G$, for some $t=\poly(\log|G|)$;\footnotemark
\item
a list of $t$ prime numbers $r_1,\ldots,r_t$;
\item 
a list of integers $\alpha_{i,j}$, for $i\in\{1,\ldots,s\}$ and $j\in \{1,\ldots,t\}$;
\item
a list of integers $\beta_{i,j}$, for $i\in\{2,\ldots,t\}$ and $j\in \{1,\ldots,i-1\}$;
\item
a list of integers $\gamma_{i,j,\ell}$, for $i\in\{2,\ldots,t\}$ and $j,\ell\in \{1,\ldots,i-1\}$.
\end{enumerate}
Let us write $H_i=\gen{h_1,\ldots,h_i}$ for each $1\le i\le t$ and $H_0=\{e\}$.\vspace{2mm}

\item
The verifier checks that the following equalities hold:
\begin{enumerate}
\item 
$g_i=h_1^{\alpha_{i,1}}\cdots h_t^{\alpha_{i,t}}$ for all $i\in\{1,\ldots,s\}$;\vspace{1mm}
\item
$h_i^{r_i}=h_1^{\beta_{i,1}}\cdots h_{i-1}^{\beta_{i,i-1}}$ for all $i\in\{2,\ldots,t\}$ and $h_1^{r_1}=e$;\vspace{1mm}
\item
$h_ih_\ell h_i^{-1}=h_1^{\gamma_{i,1,\ell}}\cdots h_{i-1}^{\gamma_{i,i-1,\ell}}$ for all $i\in\{2,\ldots,t\}$ and all $\ell\in \{1,\ldots,i-1\}$.\vspace{1mm}
\end{enumerate}
If any of these equalities fails, then the verifier aborts the protocol and outputs $\perp$.\vspace{2mm}

\item
For each $i\in\{1,\ldots,t\}$, the verifier takes a bit $s_i\in\{0,1\}$ uniformly at random and takes a random element $x_i\in H_{i-1}$ using the algorithm of Theorem \ref{th_babai} with $\varepsilon=1/2^{2n}$.\vspace{2mm}

\item
The verifier sends, for each $i\in\{1,\ldots,t\}$, the element $h_i^{s_i}x_i$.\vspace{2mm}

\item
The prover sends to the verifier bits $b_1,\ldots,b_t$ and integers $a_{i,j}$ for $i\in\{1,\ldots t\}$ and $j\in\{1,\ldots,i-1\}$.\vspace{2mm}

\item
For each $i\in\{1,\ldots,t\}$ the verifier does the following:
\begin{itemize}
\item[5.1]
If $h_i=h_1^{a_{i,1}}\cdots h_{i-1}^{a_{i,i-1}}$ then set $\ell_i=1$;
\item[5.2]
If $h_i\neq h_1^{a_{i,1}}\cdots h_{i-1}^{a_{i,i-1}}$ and $b_i=s_i$ then set $\ell_i=r_i$;
\item[5.3] 
If neither of these two conditions holds, then abort the protocol and output $\perp$.
\end{itemize}\vspace{2mm}

\item
The verifier outputs the product of the $\ell_i$'s.
\end{enumerate}
%
\end{minipage}
}
\end{center}\vspace{-4mm}
\caption{Our 3-message protocol computing the order of a solvable group.}\label{fig:protocol2}
\end{figure}

\footnotetext{Naturally, this is binary strings corresponding to the elements $h_1,\ldots,h_t$ (i.e., the oracle representations of these elements) that are actually sent, not the elements themselves. Note also that, to simplify the exposition, we are assuming that these strings do correspond to elements of $G$. To deal with a cheating prover that may send strings not corresponding to group elements, we can simply ask the prover to send a certificate of membership in $G$ for each string (such a certificate can be computed in quantum polynomial time using the algorithms of Theorem \ref{th_watrous}).}

\subsection{Analysis of the protocol}

Let us consider the correctness. In that case the prover first uses the algorithm of Theorem~\ref{th_watrous-order} to compute the order $|G|$, then factorizes it using Shor's algorithm \cite{Shor} and collects the prime factors in a set $S$.   The prover then uses the algorithm of Theorem \ref{th:classical2} using the set~$S$ as input to obtain group elements $h_1,\ldots,h_t$ and a list of integers $r_1,\ldots,r_t\in S$ such that with probability at least $1-1/\poly(|G|)$ the following two conditions hold:
\begin{enumerate}[(i)]
 \item
 $(h_1,\ldots,h_t)$ is a  polycyclic generating sequence of $G$, with $t=\poly(\log|G|)$,
\item
the order of $H_{i}/H_{i-1}$ is either 1 or $r_i$ for each $1\le i \le t$,
\end{enumerate}
where as usual we use the notation $H_i=\gen{h_1,\ldots, h_i}$ for any $i\in\{1,\ldots,t\}$ and the convention $H_0=\{e\}$.
These two conditions are equivalent to the following:
\begin{enumerate}[(a)]
\item
$H_t=G$, i.e., $g_i\in H_t$ for each $i\in\{1,\ldots,s\}$;
\item
$h_i^{r_i}\in H_{i-1}$ for each $i\in\{1,\ldots, t\}$;
\item
$H_{i-1}$ is normal in $H_i$ for any $i\in\{2,\ldots,t\}$, i.e.,
$h_ih_\ell h_i^{-1}\in H_{i-1}$ for any $\ell\in\{1,\ldots,i-1\}$.
\end{enumerate}
Thus, with probability at least $1-1/\poly(|G|)$, the prover can compute the following decompositions in quantum polynomial time using Algorithm $\Aa_1$ of Theorem \ref{th_watrous}:
\begin{itemize}
\item
a decomposition $(\alpha_{i,1}, \ldots,\alpha_{i,t})$ of $g_i$ over $H_t$, for each $i\in\{1,\ldots,s\}$;
\item
a decomposition $(\beta_{i,1}, \ldots,\beta_{i,i-1})$ of $h_i^{r_i}$ over $H_{i-1}$, for each $i\in\{2,\ldots,t\}$;
\item
a decomposition $(\gamma_{i,1,\ell}, \ldots,\gamma_{i,i-1,\ell})$ of $h_ih_\ell h_i^{-1}$ over $H_{i-1}$, for each $i\in\{2,\ldots,t\}$ and each $\ell\in \{1,\ldots,i-1\}$.
\end{itemize}
At Step 0, the prover sends all these integers, along with the elements $h_1,\ldots,h_t$ and the primes $r_1,\ldots,r_t$. All the tests performed by the verifier at Step~1 then pass. The analysis of the second part of the protocol (Steps 2-6) is then exactly the same as the analysis of the protocol of Section \ref{sec:2m}.

The soundness follows by observing that passing the tests performed by the verifier at Step 1 guarantees that Conditions (a)-(c) of the previous paragraph hold. This guarantees that Conditions (i)-(ii) hold as well, and thus the soundness analysis for the second part of the protocol (Steps 2-6) is exactly the same as the analysis of the protocol of Section \ref{sec:2m}.

\section*{Acknowledgments}
FLG was partially supported by the JSPS KAKENHI grants No.~15H01677, No.~16H01705 and No.~16H05853. TM is supported by JST PRESTO No.~JPMJPR176A, and the Grant-in-Aid for Young Scientists (B) No.~JP17K12637 of JSPS.  HN was partially supported by the JSPS KAKENHI grants No.~26247016, No.~16H01705 and No.~16K00015. YT is supported by the Program for Leading Graduate Schools: Interactive Materials Science Cadet Program.


\begin{thebibliography}{10}

\bibitem{Aaronson+07}
Scott Aaronson and Greg Kuperberg.
\newblock Quantum versus classical proofs and advice.
\newblock {\em Theory of Computing}, 3:129--157, 2007.
\newblock \href {http://dx.doi.org/10.4086/toc.2007.v003a007}
  {\path{doi:10.4086/toc.2007.v003a007}}.

\bibitem{Aharonov}
Dorit Aharonov, Michael Ben-Or, Elad Eban, and Urmila Mahadev.
\newblock Interactive proofs for quantum computations.
\newblock {\em arXiv:1704.04487}, 2017.

\bibitem{AharonovGreen}
Dorit Aharonov and Ayal Green.
\newblock A quantum inspired proof of {$P^{\#P}\subseteq IP$}.
\newblock {\em arXiv:1710.09078}, 2017.

\bibitem{AharonovVazirani}
Dorit Aharonov and Umesh Vazirani.
\newblock Is quantum mechanics falsifiable? {A} computational perspective on
  the foundations of quantum mechanics.
\newblock {\em arXiv:1206.3686}, 2012.

\bibitem{BabaiSTOC91}
L{\'{a}}szl{\'{o}} Babai.
\newblock Local expansion of vertex-transitive graphs and random generation in
  finite groups.
\newblock In {\em Proceedings of the 23rd Annual {ACM} Symposium on Theory of
  Computing}, pages 164--174, 1991.
\newblock \href {http://dx.doi.org/10.1145/103418.103440}
  {\path{doi:10.1145/103418.103440}}.

\bibitem{Babai92}
L{\'a}szl{\'o} Babai.
\newblock Bounded round interactive proofs in finite groups.
\newblock {\em SIAM Journal on Discrete Mathematics}, 5(1):88--111, 1992.
\newblock \href {http://dx.doi.org/10.1137/0405008}
  {\path{doi:10.1137/0405008}}.

\bibitem{Babai+STOC09}
L{\'{a}}szl{\'{o}} Babai, Robert Beals, and {\'{A}}kos Seress.
\newblock Polynomial-time theory of matrix groups.
\newblock In {\em Proceedings of the 41st Annual {ACM} Symposium on Theory of
  Computing}, pages 55--64, 2009.
\newblock \href {http://dx.doi.org/10.1145/1536414.1536425}
  {\path{doi:10.1145/1536414.1536425}}.

\bibitem{Babai+JCSS95}
L{\'{a}}szl{\'{o}} Babai, Gene Cooperman, Larry Finkelstein, Eugene~M. Luks,
  and {\'{A}}kos Seress.
\newblock Fast monte carlo algorithms for permutation groups.
\newblock {\em Journal of Computer and System Sciences}, 50(2):296--308, 1995.
\newblock \href {http://dx.doi.org/10.1006/jcss.1995.1024}
  {\path{doi:10.1006/jcss.1995.1024}}.

\bibitem{Babai+FOCS84}
L{\'a}szl{\'o} Babai and Endre Szemer{\'e}di.
\newblock On the complexity of matrix group problems {I}.
\newblock In {\em Proceedings of the 25th Annual Symposium on Foundations of
  Computer Science}, pages 229--240, 1984.
\newblock \href {http://dx.doi.org/10.1109/SFCS.1984.715919}
  {\path{doi:10.1109/SFCS.1984.715919}}.

\bibitem{BV}
Ethan Bernstein and Umesh Vazirani.
\newblock Quantum complexity theory.
\newblock {\em SIAM Journal on Computing}, 26:1411--1473, 1997.
\newblock \href {http://dx.doi.org/10.1137/S0097539796300921}
  {\path{doi:10.1137/S0097539796300921}}.

\bibitem{Blackburn+07}
Simon~R. Blackburn, Peter~M. Neumann, and Geetha Venkataraman.
\newblock {\em Enumeration of Finite Groups}.
\newblock Cambridge University Press, 2017.
\newblock \href {http://dx.doi.org/10.1017/CBO9780511542756}
  {\path{doi:10.1017/CBO9780511542756}}.

\bibitem{Tommaso}
Tommaso~F. Demarie, Yungkai Ouyang, and Joseph~F. Fitzsimons.
\newblock Classical verification of quantum circuits containing few basis
  changes.
\newblock {\em arXiv:1612.04914}, 2016.

\bibitem{posthoc}
Joseph~F. Fitzsimons, Michael Hajdu{\v s}ek, and Tomoyuki Morimae.
\newblock Post hoc verification of quantum computation.
\newblock {\em Physical Review Letters}, 120:040501, 2018.
\newblock \href {http://dx.doi.org/10.1103/PhysRevLett.120.040501}
  {\path{doi:10.1103/PhysRevLett.120.040501}}.

\bibitem{FK}
Joseph~F. Fitzsimons and Elham Kashefi.
\newblock Unconditionally verifiable blind computation.
\newblock {\em Phys. Rev. A}, 96:012303, 2017.
\newblock \href {http://dx.doi.org/10.1103/PhysRevA.96.012303}
  {\path{doi:10.1103/PhysRevA.96.012303}}.

\bibitem{Goldwasser+ITCS18}
Shafi Goldwasser, Ofer Grossman, and Dhiraj Holden.
\newblock Pseudo-deterministic proofs.
\newblock In {\em Proceedings of the 9th Innovations in Theoretical Computer
  Science Conference}, pages 17:1--17:18, 2018.
\newblock \href {http://dx.doi.org/10.4230/LIPIcs.ITCS.2018.17}
  {\path{doi:10.4230/LIPIcs.ITCS.2018.17}}.

\bibitem{HM}
Masahito Hayashi and Tomoyuki Morimae.
\newblock Verifiable measurement-only blind quantum computing with stabilizer
  testing.
\newblock {\em Physical Review Letters}, 115:220502, 2015.
\newblock \href {http://dx.doi.org/10.1103/PhysRevLett.115.220502}
  {\path{doi:10.1103/PhysRevLett.115.220502}}.

\bibitem{Holt+05}
Derek~F. Holt, Bettina Eick, and Eamonn~A. O'Brien.
\newblock {\em Handbook of computational group theory}.
\newblock Chapman \& Hall/CRC, 2005.
\newblock \href {http://dx.doi.org/10.1201/9781420035216}
  {\path{doi:10.1201/9781420035216}}.

\bibitem{Isaacs08}
I.~Martin Isaacs.
\newblock {\em Finite group theory}.
\newblock American Mathematical Society, 2008.

\bibitem{Ivanyos+03}
G{\'{a}}bor Ivanyos, Fr{\'{e}}d{\'{e}}ric Magniez, and Miklos Santha.
\newblock Efficient quantum algorithms for some instances of the non-abelian
  hidden subgroup problem.
\newblock {\em International Journal of Foundations of Computer Science},
  14(5):723--740, 2003.
\newblock \href {http://dx.doi.org/10.1142/S0129054103001996}
  {\path{doi:10.1142/S0129054103001996}}.

\bibitem{Ji}
Zhengfeng Ji.
\newblock Classical verification of quantum proofs.
\newblock In {\em Proceedings of the 48th Annual {ACM} symposium on Theory of
  Computing}, pages 885--898, 2016.
\newblock \href {http://dx.doi.org/10.1145/2897518.2897634}
  {\path{doi:10.1145/2897518.2897634}}.

\bibitem{Lund+92}
Carsten Lund, Lance Fortnow, Howard~J. Karloff, and Noam Nisan.
\newblock Algebraic methods for interactive proof systems.
\newblock {\em Journal of the {ACM}}, 39(4):859--868, 1992.
\newblock \href {http://dx.doi.org/10.1145/146585.146605}
  {\path{doi:10.1145/146585.146605}}.

\bibitem{Mahadev}
Urmila Mahadev.
\newblock Classical verification of quantum computations.
\newblock {\em arXiv:1804.01082}, 2018.

\bibitem{MattFourier}
Mathew McKague.
\newblock Interactive proofs with efficient quantum prover for recursive
  fourier sampling.
\newblock {\em Chicago Journal of Theoretical Computer Science}, 2012(6), 2012.
\newblock \href {http://dx.doi.org/10.4086/cjtcs.2012.006}
  {\path{doi:10.4086/cjtcs.2012.006}}.

\bibitem{MattMBQC}
Mathew McKague.
\newblock Interactive proofs for {BQP} via self-tested graph states.
\newblock {\em Theory of Computing}, 12(3):1--42, 2016.
\newblock \href {http://dx.doi.org/10.4086/toc.2016.v012a003}
  {\path{doi:10.4086/toc.2016.v012a003}}.

\bibitem{MNS}
Tomoyuki Morimae, Daniel Nagaj, and Norbert Schuch.
\newblock Quantum proofs can be verified using only single-qubit measurements.
\newblock {\em Physical Review A}, 93:022326, 2016.
\newblock \href {http://dx.doi.org/10.1103/PhysRevA.93.022326}
  {\path{doi:10.1103/PhysRevA.93.022326}}.

\bibitem{FH2}
Tomoyuki Morimae, Yuki Takeuchi, and Harumichi Nishimura.
\newblock Merlin-{Arthur} with efficient quantum {Merlin} and quantum supremacy
  for the second level of the fourier hierarchy.
\newblock {\em arXiv:1711.10605}, 2017.

\bibitem{RUV}
Ben~W. Reichardt, Falk Unger, and Umesh Vazirani.
\newblock Classical command of quantum systems.
\newblock {\em Nature}, 496:456--460, 2013.
\newblock \href {http://dx.doi.org/10.1038/nature12035}
  {\path{doi:10.1038/nature12035}}.

\bibitem{Shamir92}
Adi Shamir.
\newblock {IP} = {PSPACE}.
\newblock {\em Journal of the {ACM}}, 39(4):869--877, 1992.
\newblock \href {http://dx.doi.org/10.1145/146585.146609}
  {\path{doi:10.1145/146585.146609}}.

\bibitem{FH}
Yaoyun Shi.
\newblock Quantum and classical tradeoffs.
\newblock {\em Theoretical Computer Science}, 344:335--343, 2005.
\newblock \href {http://dx.doi.org/10.1016/j.tcs.2005.03.053}
  {\path{doi:10.1016/j.tcs.2005.03.053}}.

\bibitem{Shor}
Peter~W. Shor.
\newblock Polynomial-time algorithms for prime factorization and discrete
  logarithms on a quantum computer.
\newblock {\em {SIAM} Journal on Computing}, 26(5):1484--1509, 1997.
\newblock \href {http://dx.doi.org/10.1137/S0097539795293172}
  {\path{doi:10.1137/S0097539795293172}}.

\bibitem{Simon}
Daniel~R. Simon.
\newblock On the power of quantum computation.
\newblock {\em {SIAM} Journal on Computing}, 26(5):1474--1483, 1997.
\newblock \href {http://dx.doi.org/10.1137/S0097539796298637}
  {\path{doi:10.1137/S0097539796298637}}.

\bibitem{WatrousFOCS00}
John Watrous.
\newblock Succinct quantum proofs for properties of finite groups.
\newblock In {\em Proceedings of the 41st Annual Symposium on Foundations of
  Computer Science}, pages 537--546, 2000.
\newblock \href {http://dx.doi.org/10.1109/SFCS.2000.892141}
  {\path{doi:10.1109/SFCS.2000.892141}}.

\bibitem{Watrous01}
John Watrous.
\newblock Quantum algorithms for solvable groups.
\newblock In {\em Proceedings of the 33rd Annual {ACM} Symposium on Theory of
  Computing}, pages 60--67, 2001.
\newblock \href {http://dx.doi.org/10.1145/380752.380759}
  {\path{doi:10.1145/380752.380759}}.

\bibitem{IBM}
{https://www-03.ibm.com/press/us/en/pressrelease/52403.wss}.

\end{thebibliography}

\end{document}